\pdfoutput=1
\documentclass[letterpaper, 10 pt, conference]{ieeeconf}  

\IEEEoverridecommandlockouts                              

\usepackage{graphics} 
\usepackage[cmex10]{amsmath}

\usepackage{amssymb,amsthm,amsfonts}
\usepackage[hyphens]{url}
\usepackage[american]{babel}

\usepackage{graphicx}
\usepackage[caption=false,font=footnotesize]{subfig}
\usepackage{multirow}
\usepackage{algorithm,algorithmic}
\makeatletter
\renewcommand{\ALG@name}{Model}
\makeatother
\usepackage{array}
\newcolumntype{K}[1]{>{\centering\arraybackslash}p{#1}}
\usepackage{bm}
\usepackage{xcolor}
\usepackage{enumerate}
\usepackage{stfloats}
\usepackage{placeins}
\usepackage{microtype}
\let\labelindent\relax
\usepackage{enumitem}

\definecolor{myred}{RGB}{202,0,32}
\definecolor{myorange}{RGB}{244,165,130}
\definecolor{myviolet}{RGB}{194,165,207}
\definecolor{mycyan}{RGB}{146,197,222}
\definecolor{myblue}{RGB}{5,113,176}
\definecolor{mygreen}{RGB}{127,191,123}
\definecolor{mytile}{RGB}{27,120,55}

\newtheorem{lemma}{Lemma}[section]

\newtheorem{prop}[lemma]{Proposition}

\newtheorem{thm}[lemma]{Theorem}

\theoremstyle{remark}
\theoremstyle{remark}
\theoremstyle{definition}

\newcommand{\set}[1]{\left\{#1\right\}}
\newcommand{\abs}[1]{\left|#1\right|}

\DeclareMathOperator{\diag}{diag}

\newcommand{\R}{\mathbb R}

\newcommand{\calE}{\mathcal{E}}

\newcommand{\calG}{\mathcal{G}}

\newcommand{\calN}{\mathcal{N}}

\newcommand{\pp}{{\bm{p}}}
\newcommand{\ff}{{\bm{f}}}
\newcommand{\bt}{{\bm{\theta}}}
\newcommand{\ee}{{\bm{e}}}
\newcommand{\BB}{{\bm{B}}}
\newcommand{\CC}{{\bm{C}}}
\newcommand{\LL}{{\bm{L}}}

\usepackage[normalem]{ulem}

\title{\LARGE \bf
Interface Networks for Failure Localization in Power Systems}

\author{Chen Liang$^{1}$, Alessandro Zocca$^{2}$, Steven H. Low$^{1}$ and Adam Wierman$^{1}$
\thanks{$^{1}$Chen Liang, Steven H. Low, and Adam Wierman are with the Department of Computing and Mathematical Sciences, California Institute of Techonlogy, Pasadena, CA, USA. Email: {\tt\small \{cliang2, slow, adamw\}@caltech.edu}}%
\thanks{$^{2}$Alessandro Zocca is with the Department of Mathematics, Vrije Universiteit, Amsterdam, Netherlands. Email: {\tt\small a.zocca@vu.nl}}%
}

\begin{document}

\maketitle

\begin{abstract}
	Transmission power systems usually consist of interconnected sub-grids that are operated relatively independently. When a failure happens, it is desirable to localize its impact within the sub-grid where the failure occurs. This paper introduces three interface networks to connect sub-grids, achieving better failure localization  while maintaining robust network connectivity. The proposed interface networks are validated with numerical experiments on the IEEE 118-bus test network under both DC and AC power flow models.
\end{abstract}

\section{Introduction}\label{sec:intro}


An interconnected power system comprises sub-grids that are usually individually managed by independent system operators (ISOs). It is desirable to localize failure impact within the sub-grid where the failure happens while leaving other sub-grids unaffected. On the other hand, transmission line failures are known to propagate non-locally~\cite{hines2017cascading,guo2017critical}. Historical data shows that successive failures in large cascades can be far away from the preceding failure, both geometrically and topologically~\cite{kinney2005modeling}.

Considerable attention in recent years has been given to the task of understanding control policies and network structural properties that can localize the impact of failures.  Papers have focused on active control actions such as load shedding~\cite{faranda2007load,saffarian2010enhancement} and controlled islanding~\cite{ding2012two,esmaeilian2016prevention} to prevent large-scale blackouts. Researchers also investigate the  relation between  failure propagation and  topological structures of power networks~\cite{cetinay2016topological,kaiser2021topological}.
It is proven that non-cut failures are localized if sub-grids are connected in a tree structure and that, if sub-grids are connected by multiple lines, failures cannot be completely localized~\cite{guo2021line1}. This suggests switching off certain transmission lines in order to leave only one line between each pair of sub-grids~\cite{guo2018failure, zocca2021, lan2021a}.

However, maintaining a tree structure at the sub-grid level is at odds with the standard approach to reliability because it creates single points of failure. Further, a tree-connected power network significantly reduces the power transmission capacity between the sub-grids, increasing the cost of power dispatch. Instead, more traditionally, it is desirable to have multiple lines between sub-grids so as to ensure there is no single-point vulnerabilities and to increase the power transmission capacity.  

These contrasting views lead to an important open question: \emph{Is it possible to provably localize failures within sub-grids without creating a single point of failure? }


\textbf{Contributions.} We study three interface networks that connect the sub-grids in a way that achieves failure localization and robust connectivity between the sub-grids.

Specifically, we consider a power network consisting of two sub-grids and propose three alternative \textit{interface networks} to connect them, as shown in Fig.~\ref{fig_sim}. We use line outage distribution factors (LODF) as the metric to quantify failure localization and prove, in Theorems \ref{thm:series} and \ref{thm:parallel},
that the LODFs are guaranteed to decrease if the sub-grids are connected by a $2\times2$ series (Fig.~\ref{fig_second_case}) or parallel (Fig.~\ref{fig_third_case}) interface network. We further provide an upper bound for the LODF if the sub-grids are connected by a $2\times2$ complete bipartite network (Fig.~\ref{fig_forth_case}) in Theorem~\ref{thm:bipartite}. 
By carefully designing the line susceptances of the interface network, the complete bipartite interface network can completely eliminate  failure propagation to other sub-grids while keeping the impact on surviving lines in the same sub-grid unchanged. We validate the efficacy of the proposed interface networks on the IEEE 118-bus test network under both DC and AC power flow models. All three interface networks decrease the LODF for lines in different sub-grids, with the complete bipartite network achieving the best localization. 

\textbf{Related Work.} There have been extensive efforts toward understanding the localizability of failures in power systems to network topological structures. Most papers focus on summarizing empirical results.  For example, \cite{soltan2015analysis} observes the LODF decreases as the distance from the tripped line increases, and \cite{strake2019non} defines another distance metric that better captures such decay. There are only a few papers presenting analytical results on failure localization. A topological representation for the LODF is proposed in \cite{guo2017monotonicity} and the authors further prove that LODF is zero if the sub-grids are connected in a tree structure in \cite{guo2021line1}. The failure localization of tree partitioning has been proposed in~\cite{bialek2021tree} to replace controlled islanding as a defense mechanism to arrest cascading failure.
However, a power network with tree-connected sub-grids is less practical as it creates a single point of failure. In \cite{kaiser2021network}, authors propose to connect the sub-grids by a complete bipartite interface, the network isolator, to suppress the failure spreading. However, they require the adjacency matrix (weighted by line susceptance) of the interface network to be exactly rank-1 which can be difficult to satisfy in practice. To the best of our knowledge, this paper is the first to mathematically characterize the LODF with sub-grids connected by interface networks beyond the rank-1 setting.


\section{Power Redistribution}\label{sec:pre}
To begin, we introduce the linearized DC power flow model and then illustrate how power flows redistribute in the network after line failures. We further show how to decompose the power transfer distribution factor and derive their monotonicity property in terms of line susceptances. 

\subsection{DC Power Flow}

We model the power grid as a directed graph $\calG = (\calN, \calE)$ with a set $\calN=\set{1,2,\dots,n}$ of $n$ buses and a set $\calE=\set{e_1,\dots, e_m} \subseteq \calN \times \calN$ of $m$ transmission lines connecting the buses. An arbitrary direction is assigned to each transmission line, and $(i,j)$ represents the transmission line from bus $i$ to bus $j$. We assume the lines are purely reactive and characterized by their susceptances, which we collect in the susceptance matrix $\BB:=\diag{(b_1, \dots, b_m)}$.

Let $\pp, \bt\in \R^n$ denote the power injection and phase angle at each bus, and let $\ff\in\R^m$ denote the power flow along every transmission line. The widely used DC power flow equations~\cite{stott2009dc} can be written in the following matrix form:
\begin{subequations}\label{eqn:dc_model}
	\begin{IEEEeqnarray}{rCl}
		\pp&=&\CC\ff ,\label{eqn:dc_model.1}\\
		\ff&=&\BB\CC^T\bt, \label{eqn:dc_model.2}
	\end{IEEEeqnarray}
\end{subequations}
where $\CC\in\R^{n\times m}$ is the incidence matrix:
$$
C_{ie}=\begin{cases}
  1 & \text{if bus }i\text{ is the source of line } e,\\
  -1 & \text{if bus }i\text{ is the destination of line } e,\\
  0 &\text{otherwise.}
\end{cases}
$$
If the  network is connected, the power injections must be balanced, i.e., $\sum_{i\in\calN} p_i =0$. Defining the Laplacian matrix of the network as $\LL:=\CC\BB\CC^T$, we can uniquely determine the line flows in terms of the power injections:
\begin{equation}
    \ff = \BB \CC^T \LL^\dagger \pp,
\end{equation}
where $(\cdot)^\dagger$ denotes the Moore-Penrose inverse.

\subsection{Power Redistribution After Line Failures}
When a line failure occurs, the power will redistribute over the post-contingency network, and line flows can both increase or decrease, sometimes even reversing their directions. In the power systems literature two sensitivity factors, the power transfer (PTDF) and the line outage (LODF) distribution factors are commonly used to compute the post-contingency line flows~\cite{guler2007generalized,guo2009direct}. It should be noted that these sensitivity factors are independent of the power injections and transmission line capacities.

Specifically, the PTDF $D_{e,\hat i \hat j}$ is the relative flow change over line $e=(i,j)$ when a unit power is injected at bus $\hat i$ and withdrawn from bus $\hat j$. The LODF $K_{e,\hat e}$ is the relative flow change over line $e = (i,j)$ when line $\hat e=(\hat i, \hat j)$ is tripped. They are given by:
\begin{subequations}
    \begin{IEEEeqnarray}{rCl}
        D_{e,\hat i \hat j} & = & b_e (\ee_i - \ee_j)^T \LL^\dagger (\ee_{\hat i} - \ee_{\hat j}), \label{eqn:PTDF.1}\\
        K_{e,\hat e} & = &  \frac{b_e(\ee_i - \ee_j)^T \LL^\dagger (\ee_{\hat i} - \ee_{\hat j})}{1 - b_{\hat e} (\ee_{\hat i} - \ee_{\hat j})^T \LL^\dagger (\ee_{\hat i} - \ee_{\hat j})},\label{eqn:LODF}
    \end{IEEEeqnarray}
\end{subequations}
where $\{\ee_k\}_{k=1,\dots,n}$ is the standard vector basis. It is known that the PTDF and LODF can be related as follows~\cite{wood2013power}:
$$
K_{e,\hat e} = \frac{D_{e,\hat i \hat j}}{1 - D_{\hat e,\hat i \hat j}}.
$$
This expression suggests that the power redistribution after line failures can be emulated by introducing fictitious injections over the pre-contingency network~\cite{guler2007generalized}. In fact, the power redistribution can be analyzed over the post-contingency network as well. The following lemma, which relates the LODF for the pre-contingency network and the PTDF for the post-contingency network, provides an alternative perspective to study the impact of transmission line failures.

\begin{lemma} \label{lem:LODF_PTDF}
Consider a network $\calG=(\calN, \calE)$ and a non-bridge transmission line $\hat e$ failure\footnote{A non-bridge line is a transmission line whose deletion does not increase the network's number of connected components. Otherwise it is a bridge.}. Let $K_{e, \hat e}$ denote the LODF for the pre-contingency network $\calG$, and let $\tilde{D}_{e,\hat i \hat j}$ denote the PTDF for the post-contingency network $\tilde{\calG}=(\calN, \calE \setminus \hat e)$. We have
$
K_{e, \hat e} = \tilde{D}_{e,\hat i \hat j}.
$
\end{lemma}
\begin{proof}
Without loss of generality, decompose the matrices $\CC = [\CC_{\hat e}, \CC_{-\hat e}]$ and $\BB = \diag(b_{\hat e}, b_l), l \neq \hat e$ correponding to the tripped line $\hat e$ and the surviving lines $-\hat e:=\calE \setminus \hat e$. For clarity, we use $\tilde{(\cdot)}$ to denote all variables related to the post-contingency network. The DC power flow equations for the pre- and post-contingency networks thus rewrite as:
\begin{subequations}
	\begin{IEEEeqnarray}{rCl}
		\pp & = & \CC \BB \CC^T \bt = b_{\hat e} \CC_{\hat e} \CC_{\hat e}^T\bt +  \CC_{-\hat e} B_{-\hat e} \CC_{-\hat e}^T\bt, \label{eqn:DC_pre_post.1}\\
		\pp & = & \CC_{-\hat e} \BB_{-\hat e} \CC_{-\hat e}^T \tilde{\bt}.  \label{eqn:DC_pre_post.2}
	\end{IEEEeqnarray}
\end{subequations}
Subtracting \eqref{eqn:DC_pre_post.2} from \eqref{eqn:DC_pre_post.1}, we get
$$
\CC_{-\hat e} \BB_{-\hat e} \CC_{-\hat e}^T (\tilde{\bt} - \bt) = b_{\hat e} \CC_{\hat e} \CC_{\hat e}^T\bt = \CC_{\hat e} f_{\hat e}.
$$
By definition, the LODF can be computed as
$$
K_{e,\hat e} = \frac{\tilde{f}_e -f_e}{f_{\hat e}} = \frac{ b_e \CC_{e}^T (\tilde{\bt} - \bt)} {f_{\hat e}} = \frac{b_e\CC_e^T \tilde{L}^\dagger \CC_{\hat e} f_{\hat e}}{f_{\hat e}} = \tilde{D}_{e,\hat i \hat j}.
$$%
\end{proof}%
We remark that the above result holds also in the case of multiple line failures. Consider a set $\hat E$ of lines that are simultaneously disconnected and suppose the post-contingency network remains connected. The generalized LODF $K_{e,\hat e}^{\hat E}$ ~\cite{guler2007generalized} is defined as the sensitivity of relative flow change over the surviving line $e \in \calE \setminus \hat{E}$ with respect to a tripped line $\hat e = (\hat i, \hat j) \in \hat{E}$. It equals the PTDF for line $e$ with the pair of buses $\hat i, \hat j$ of the post-contingency network.

Lemma~\ref{lem:LODF_PTDF} suggests that the impact of a transmission line $\hat e=(\hat i, \hat j)$ failure is equivalent to the power flows when the pre-contingency flow $f_{\hat e}$ is injected at bus $\hat i$ and withdrawn from bus $\hat j$ over the post-contingency network. This post-contingency perspective allows us to convert the calculation of LODF into the calculation of PTDF, relating the failure impact directly to the network topology.

\subsection{Decomposition of the PTDF}
A power grid usually consists of several interconnected sub-grids and it is of interest to decompose the calculation of PTDF accordingly. In this section, we introduce such a decomposition for certain network structures. Specifically, suppose a connected network $\calG=(\calN,\calE)$ can be decomposed into two sub-grids: $\calG_1=(\calN_1, \calE_1)$ and $\calG_2=(\calN_2, \calE_2)$ such that:
\begin{itemize}
    \item The line sets do not overlap: $\calE_1 \cap \calE_2 = \emptyset$, $\calE_1 \cup \calE_2 = \calE$;
    \item The bus sets overlap with  only 2 buses: $\calN_1 \cap \calN_2 = \set{s, t}$, $\calN_1 \cup \calN_2 = \calN$.
\end{itemize}
Given a pair of buses $i, j$ of the network $\calG$ (not necessarily adjacent to each other), we define 
\emph{effective susceptance} between $i,j$ to be
\begin{equation}
    b_{ij}^{(e)} = \frac{1}{(\ee_{i} - \ee{j})^T \LL^\dagger (\ee_{i} - \ee_{j})}.
\end{equation}
The effective susceptance summarizes the network effect between a pair of buses by a single line ~\cite{ghosh2008minimizing}.

The following proposition demonstrates how $D_{e,\hat i \hat j}$, the PTDF for line $e$ and a pair of buses $\hat i,\hat j$ in different sub-grids, can be decomposed.
\begin{prop} \label{prop:PTDF_decomp}
Consider a network $\calG$ and its decomposition $\calG_1$, $\calG_2$. Let $\hat \calG_1 = (\calN_1, \calE_1 \cup (s,t))$ be a graph by adding a fictitious line
$(s,t)$ to the sub-grid $\calG_1$, with susceptance equaling the effective susceptance between buses $s,t$ of the sub-grid $\calG_2$:
$$b_{st}^{(e)}=\frac{1}{(\ee_s-\ee_t)^T \LL_2^\dagger (\ee_s-\ee_t)},$$
where $\LL_2$ is the Laplacian matrix of $\calG_2$.
For any pair of buses $\hat i,\hat j \in \calN_1$ and any line $e\in\calE_2$, the PTDF $D_{e,\hat i \hat j}$ can be computed as:
$$
D_{e,\hat i \hat j} = \hat D_{(s,t), \hat i \hat j} \cdot \bar{D}_{e, st},
$$
where $\hat D_{(s,t), \hat i \hat j}$ is the PTDF for the fictitious line $(s,t)$ and the pair of buses $\hat i, \hat j$ of $\hat \calG_1$, and $\bar{D}_{e, st}$ is the PTDF for line $e$ and the pair of buses $s,t$ of $\calG_2$.
\end{prop}
\begin{proof}[Proof (sketch)]
Since a DC power network is a linear network and the sub-grids are only joined by buses $s$ and $t$, the effect of $\calG_2$ can be equivalently represented as a fictitious transmission line $(s,t)$ with the effective susceptance between buses $s,t$ of the sub-grid $\calG_2$. Using Kron reduction~\cite{dorfler2013kron}, we can decompose the PTDF as above.
\end{proof}
We remark that this result is in fact a special case of the Kron reduction~\cite{dorfler2013kron} for linear networks.


\subsection{Monotonicity of the PTDF}


The next result describes the dependence of the PTDF for a line $e$ on its susceptance $b_e$ and network topology.
\begin{prop}\label{prop:mono_PTDF}
	Consider a connected network $\calG$. For any line $e$ and any pair of buses $\hat i,\hat j$, the absolute value of PTDF $D_{e,\hat i\hat j}$ can be expressed as:
	\begin{equation}
	|D_{e,\hat i\hat j}| = \frac{T_1 b_e}{T_2 b_e + T_3}, \label{eqn:PTDF}
	\end{equation}
	where $T_i \geq 0$ is a constant independent of the susceptance $b_e$ for $i=1,2,3$. Moreover, $T_3=0$ if and only if $e$ is a bridge of $\calG$. 
\end{prop}
\begin{proof}
	The PTDF can be computed as a quotient of different spanning trees of the network, as shown in Theorem 4 in~\cite{guo2021line1}. Specifically, the numerator involves a subset of spanning trees that must pass through line $e$. The denominator involves all spanning trees, including those that pass through line $e$, accounted for in the term $T_2b_e$, and those that do not, giving rise to the term $T_3$.
    Therefore, $T_3=0$ if and only if $e$ is a bridge.
\end{proof}
This proves that the absolute value of PTDF for a line is monotonically increasing in its susceptance if the line is not a bridge. This monotonicity result is aligned with the intuition that lines with larger admittances (thus smaller impedances) tend to ``attract'' more power to flow through.

\section{Failure Localization}\label{sec:localization}
\begin{figure}[!t]
\centering
\subfloat[Original network]{\includegraphics[width=0.43\columnwidth]{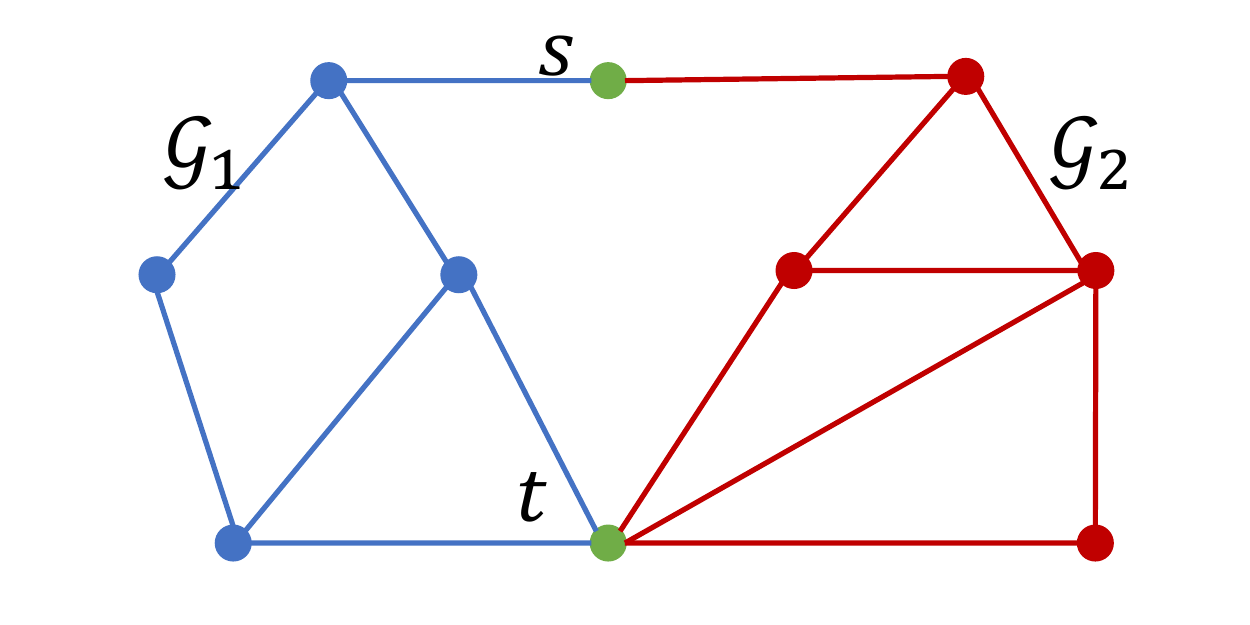} \label{fig_first_case}}
\hfil
\subfloat[Series network]{\includegraphics[width=0.43\columnwidth]{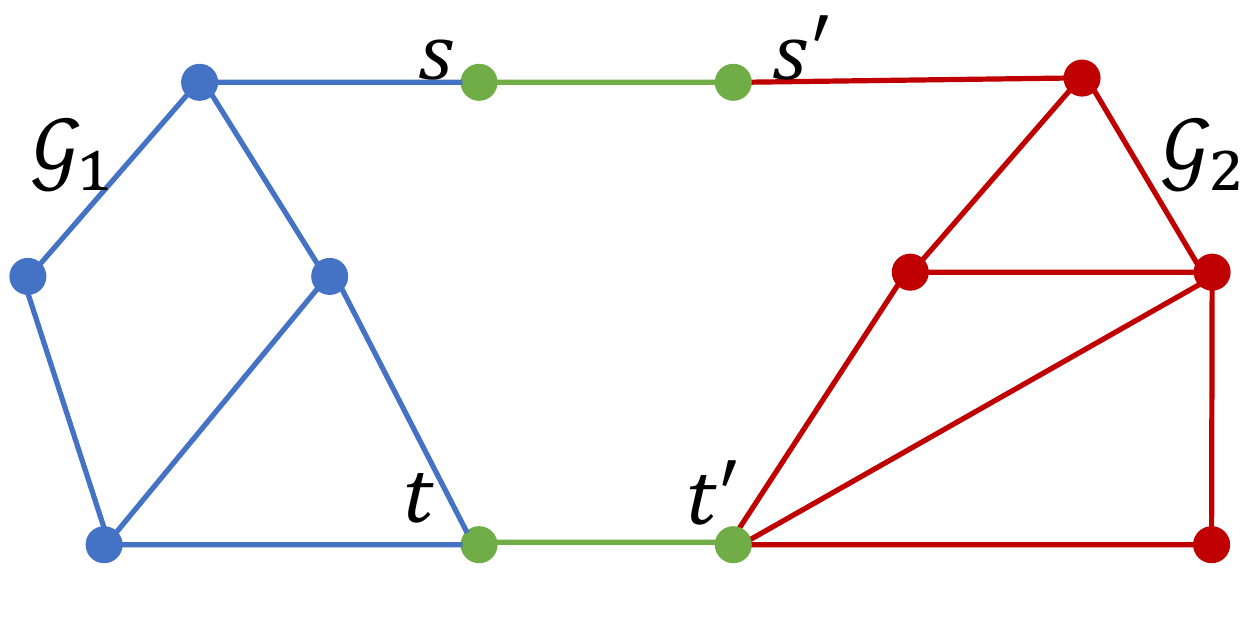} \label{fig_second_case}}
\\
\subfloat[Parallel network]{\includegraphics[width=0.43\columnwidth]{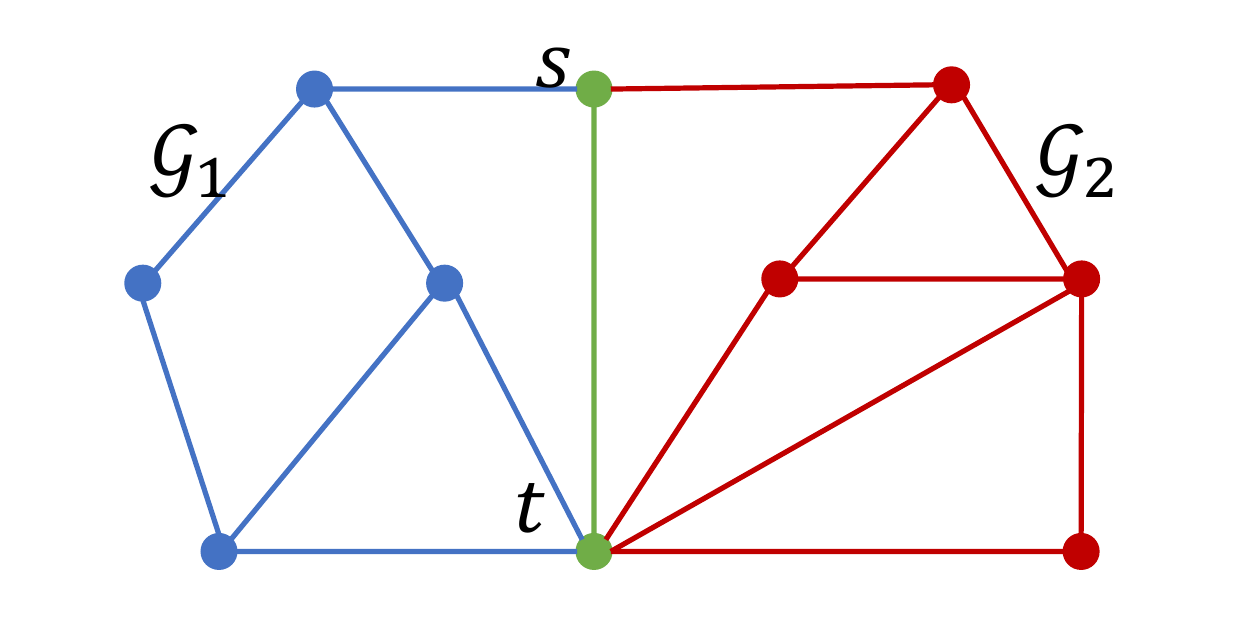} \label{fig_third_case}}
\hfil
\subfloat[Complete bipartite network]{\includegraphics[width=0.43\columnwidth]{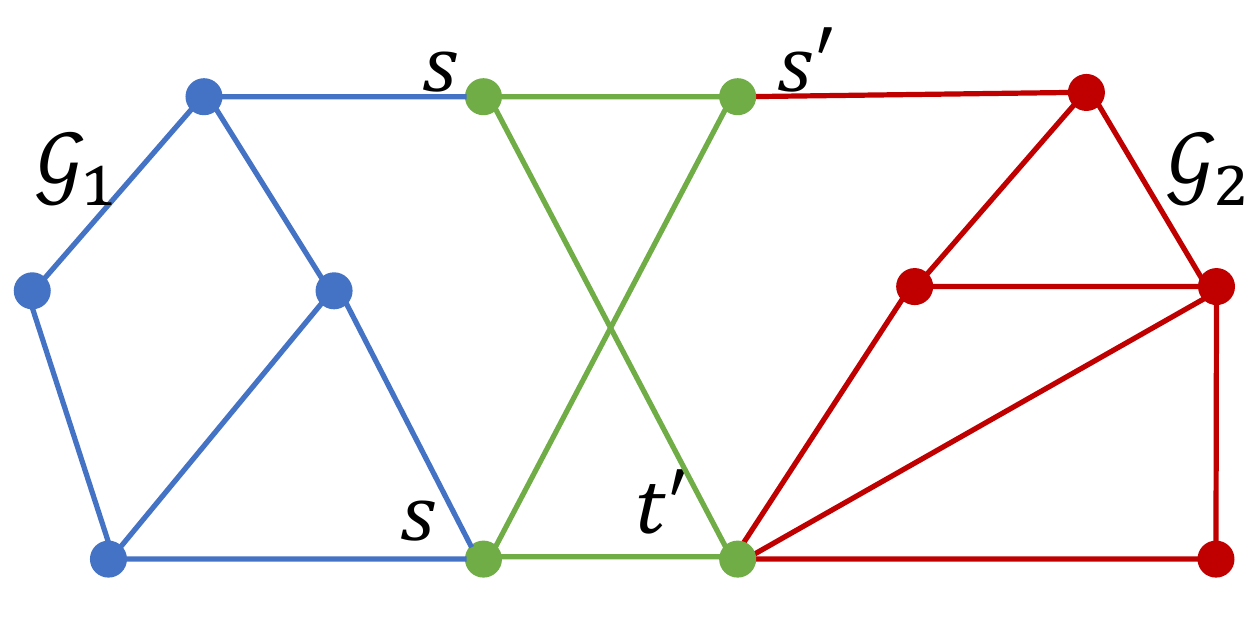} \label{fig_forth_case}}
\caption{Two sub-grids (a) are interconnected by (b) series, (c) parallel and (d) complete bipartite interface networks. 
}
\label{fig_sim}
\end{figure}


The sub-grids that make up an interconnected power system operate relatively independently and it is desirable to localize the impact of failure within the sub-grid to prevent large-scale blackouts.
Failures in power systems, however, are known to propagate non-locally. Real-world data shows that a transmission line failure can lead to another line failure far away from the initial failure~\cite{kinney2005modeling}. A recent study investigates the block decomposition of the power networks and demonstrates that non-cut failures are localized if the sub-grids are connected in a tree structure~\cite{guo2021line1}. In practice, however, designing a power system with tree-connected sub-grids creates bridges and thus introduces single-point vulnerabilities. Therefore, it is crucial to localize failures while maintaining connectivity of the grid.

In this section we consider a power network with two interconnected sub-grids $\calG_1, \calG_2$ joined by two buses $s$ and $t$ and propose three interface networks, as shown in Fig.~\ref{fig_sim}. In contrast to the tree-connected sub-grids proposed in~\cite{guo2018failure}, our design does not decrease the connectivity or introduce any single point of failure into the original network. 
We show that all three interface networks can achieve failure localization by carefully designing the susceptances. Note that the paper focuses on the interface networks of the power grid where sub-grids are joined by two buses. Larger interface networks require special topological structures of the sub-grids to guarantee failure localization. For this reason we leave this as a challenging topic for future work.

To quantify the benefit of the interface networks in Fig.~\ref{fig_sim}, we compare the LODF $K_{e,\hat e}$ of the original network and that of the modified network with various interface networks. The tripped line $\hat e$ and the monitored line $e$ are in different sub-grids. Without loss of generality, we assume that $\hat e\in\calG_1$ and $e\in\calG_2$. We assume buses $s,t$ are not directly connected (i.e., not adjacent to each other) in the original network to simplify our discussion. We use the superscript $(\cdot)^{(m)}$ to denote variables corresponding to the modified network.

\subsection{Series Interface Network}
We first introduce the $2\times2$ \emph{series} interface network where we split the buses $s,t$ and connect $(s,s')$ and $(t,t')$ as additional transmission lines, as shown in Fig.~\ref{fig_second_case}. Intuitively, the series interface network increases the topological distance between the tripped line $\hat e$ and the monitored line $e$. It is thus likely to reduce the failure impact across the sub-grids. As characterized by the following theorem, the LODF $K_{e,\hat e}^{(m)}$ for the modified grid with a series interface network is guaranteed to decrease under mild conditions.


\begin{thm}\label{thm:series}
If $\calG_1$ and $\calG_2$ are connected by a series network, then $|K^{(m)}_{e,\hat e}| \leq |K_{e,\hat e}|$, with equality if and only if there is no path connecting buses $s$ and $t$ in the post-contingency network of $\calG_1$. 
\end{thm}

\begin{proof}
	With Lemma~\ref{lem:LODF_PTDF} and Proposition~\ref{prop:PTDF_decomp}, we can write the LODF for the original network as:
	$$
	K_{e,\hat e} = \hat{D}_{(s,t), \hat i \hat j} \cdot \bar{D}_{e, st}.
	$$
	We use $\hat{D}_{(s,t), \hat i \hat j}$ to denote the PTDF for the fictitious line $(s,t)$ of the \emph{post}-contingency sub-grid of $\hat{\calG}_1$, with susceptance being the effective susceptance $b_{st}^{(e)}$ of the sub-grid $\calG_2$. $\bar{D}_{e, st}$ represents the PTDF for line $e$ of the sub-grid $\calG_2$.

	For the modified network $\calG^{(m)}$, let $\calG_1^{(m)}=\calG_1$ and $\calG_2^{(m)} = (\calN_2 \cup \{s', t'\}, \calE_2 \cup \{(s,s'), (t,t')\})$. The LODF can be decomposed similarly as:
	$$
	K_{e,\hat e}^{(m)} = \hat{D}^{(m)}_{(s,t), \hat i \hat j} \cdot \bar{D}^{(m)}_{e, st}.
	$$
	Note that $\bar{D}^{(m)}_{e, st} =  \bar{D}_{e, st}$ since line $(s,s')$ and $(t,t')$ are bridges of $\calG_2^{(m)}$. The effective susceptance for the fictitious line $b_{st}^{(m)} = (1/b_{ss'} + 1/b_{tt'} + 1/b_{st}^{(e})^{-1} < b_{st}^{(e)}$, so we have $|\hat{D}^{(m)}_{(s,t)} | \leq |\hat{D}_{(s,t)}| $ from Proposition~\ref{prop:mono_PTDF}. Therefore we conclude $|K_{e,\hat e}^{(m)}| \leq |K_{e,\hat e}|$, with equality if and only if the fictitious line $(s,t)$ is a bridge in the post-contingency network of $\hat{\calG}_1$.
\end{proof}

When no path in the post-contingency network of $\calG_1$ connects buses $s,t$, the LODF remains the same as that of the original network. In particular, the LODF will only depend on the structure of $\calG_2$ and equal the PTDF of the monitored line $e$ for the buses $s',t'$ in $\calG_2$.
Otherwise, the LODF strictly decreases compared with that of the original network. 

We remark that many empirical studies show that the LODF decreases as the distance from the initial failure increases~\cite{soltan2015analysis,kaiser2021topological}. Theorem~\ref{thm:series} provides theoretical support for such observations. Furthermore, the LODF $K^{(m)}_{e,\hat e}$ is a non-decreasing function in the susceptance of lines $(s,s')$ and $(t,t')$. Therefore, we can design the series interface network to achieve different levels of failure localizability. 

\subsection{Parallel Interface Network}
We now consider the \emph{parallel} interface network where we connect the buses $s$ and $t$, as shown in Fig.~\ref{fig_third_case}. Effectively, the line $(s,t)$ provides an alternative path to redistribute  power without passing through the other sub-grid $\calG_2$. Therefore we expect the line failures to be less impactful on the other sub-grid. Indeed, the following theorem shows that the LODF is guaranteed to decrease after connecting buses $s,t$. 

\begin{thm}\label{thm:parallel}
If $\calG_1$ and $\calG_2$ are connected by a parallel network, then $|K^{(m)}_{e,\hat e}| < |K_{e,\hat e}|$. 
\end{thm}
\begin{proof}
Similar to the proof of Theorem~\ref{thm:series}, we can write the LODF for the original network and modified network as:
$$
K_{e,\hat e} = \hat{D}_{(s,t), \hat i \hat j} \cdot \bar{D}_{e, st}, \ 
K_{e,\hat e}^{(m)} = \hat{D}^{(m)}_{(s,t), \hat i \hat j} \cdot \bar{D}^{(m)}_{e, st},
$$
where we define $\calG_1^{(m)}=\calG_1, \calG_2^{(m)} = (\calN_2, \calE_2 \cup \{(s,t)\})$. From Proposition~\ref{prop:mono_PTDF}, we have
$$
 |\hat{D}_{(s,t), \hat i \hat j}| = \frac{T_1 b_{st}^{(e)}}{T_2 b_{st}^{(e)}+T_3}, \ 
  |\hat{D}^{(m)}_{(s,t), \hat i \hat j}| = \frac{T_1 b^{(m)}_{st}}{T_2 b^{(m)}_{st}+T_3},
$$
where $b^{(m)}_{st} = x + b_{st}^{(e)}$ with $x$ being the susceptance of the parallel line $(s,t)$. On the other hand, a simple circuit analysis shows that $\bar{D}^{(m)}_{e, st} = \frac{b_{st}^{(e)}}{b_{st}^{(m)}}\bar{D}_{e, st}$. 

Therefore, we can conclude that
\begin{subequations}
\begin{IEEEeqnarray*}{rCl}
	|K_{e,\hat e}^{(m)}| &=& \frac{T_1 b^{(m)}_{st}}{T_2 b^{(m)}_{st}+T_3} \cdot \frac{b_{st}^{(e)}}{b_{st}^{(m)}}|\bar{D}_{e, st}| \\
	&=& \frac{T_1 b_{st}^{(e)}}{T_2 b^{(m)}_{st}+T_3} |\bar{D}_{e, st}| \\
	&<& \frac{T_1 b_{st}^{(e)}}{T_2 b_{st}^{(e)}+T_3} |\bar{D}_{e, st}| \ \ = \ \ |K_{e,\hat e}|.
\end{IEEEeqnarray*}
\end{subequations}
\end{proof}
We remark that the LODF $K_{e,\hat e}$ is monotonically decreasing in the susceptance of the parallel line $(s,t)$. We can thus increase the susceptance of line $(s,t)$ to improve the failure localizability. On the other hand, the LODF for the parallel line $(s,t)$ increases as the susceptance increases according to Proposition~\ref{prop:mono_PTDF}. Thus we need to systematically design the susceptance of the line $(s,t)$.

\subsection{Complete Bipartite Network}
We now introduce the $2\times 2$ \emph{complete bipartite} interface network with two buses on each side, where we split the buses $s, t$ and connect $(s,s')$, $(s, t')$, $(t,s')$ and $(t,t')$ respectively. This design is similar to the Wheatstone bridge in circuit analysis literature. We show in the following theorem that the LODF for lines across sub-grids can be upper bounded. In particular, the impact of failures can be completely eliminated under the condition $b_{ss'}b_{tt'} = b_{st'}b_{ts'}$, where $b_{pq}$ denotes the susceptance of line $(p,q)$. We remark that this specific interface network has been proposed in~\cite{kaiser2021network} as the \emph{network isolator} and shown to provide localization if a rank-1 condition holds on the weighted adjacency matrix of the interface network. The rank-1 condition is equivalent to $b_{ss'}b_{tt'} = b_{st'}b_{ts'}$ for the $2\times2$ complete bipartite network.
Our result generalizes the failure localization properties of a network isolator to the case in which the rank-1 condition does not hold for the four-node bipartite network.

\begin{thm} \label{thm:bipartite}
	If $\calG_1$ and $\calG_2$ are connected by a complete bipartite network, then $|K^{(m)}_{e, \hat e}| \leq \frac{\abs{b_{ss'}b_{tt'} - b_{st'}b_{ts'}}}{(b_{ss'}+b_{st'})(b_{tt'}+b_{ts'})}$ where $b_{pq}$ is the susceptance for line $(p,q)$. In particular, if $b_{ss'}b_{tt'} = b_{st'}b_{ts'}$, then $K^{(m)}_{e,\hat e}=0$.
\end{thm}

\begin{proof}
We have 
$$
K_{e,\hat e}^{(m)} = \hat{D}^{(m)}_{(s,t), \hat i \hat j} \cdot \bar{D}^{(m)}_{e, st},
$$
where $\calG_1^{(m)}=\calG_1$ and $\calG_2^{(m)} = (\calN_2 \cup \{s', t'\}, \calE_2 \cup \{(s,s'),(t,t'),(s,t'),(t,s')\})$. 
Since the PTDF is guaranteed to be within $[-1,1]$, we first bound $|\hat{D}^{(m)}_{(s,t), \hat i \hat j}|$ by 1 and focus on the second term.
Moreover, we can further decompose $\calG_2$ into $\calG_2^1 = (\{s,s',t,t'\},\{(s,s'),(t,t'),(s,t'),(t,s')\})$ and $\calG_2^2 = (\calN_2 \cup \{s',t'\}, \calE_2)$. Therefore, we have
$$
|K_{e,\hat e}^{(m)}|  \leq |\bar{D}^{(m)}_{e, st}| = |\hat{D}^{1}_{(s',t'),st}| \cdot |\bar{D}^{2}_{e, s't'}| \leq |\hat{D}^{1}_{(s',t'),st}|.
$$
Now all we need is to provide an upper bound for the right hand side, which is the PTDF for the fictitious line $(s',t')$ with effective susceptance $b^{(e)}$ for $\hat{\calG}_2^1 =  (\{s,s',t,t'\},\{(s,s'),(t,t'),(s,t'),(t,s'),(s',t')\})$. We can compute the PTDF as in~\eqref{eqn:PTDF.1}:
\begin{IEEEeqnarray*}{lll}
	\hat{D}^{1}_{(s',t'),st} &=& b^{(e)}(b_{ss'}b_{tt'}-b_{st'}b_{ts'}) / \\
	& & [(b_{ss'}b_{st'}b_{ts'} +b_{ss'}b_{st'}b_{tt'} +  b_{ss'}b_{ts'}b_{tt'} +\\
	&&b_{st'}b_{ts'} b_{tt'}) + (b_{ss'}b_{st'})(b_{ts'}+b_{tt'})b^{(e)} ] 
\end{IEEEeqnarray*}
Therefore, we conclude an upper bound for the LODF:
$$
|K_{e,\hat e}^{(m)} | \leq |\hat{D}^{1}_{(s',t'),st}| \leq \frac{\abs{b_{ss'}b_{tt'} - b_{st'}b_{ts'}}}{(b_{ss'}+b_{st'})(b_{tt'}+b_{ts'})}.
$$
\end{proof}

Note that the bound depends only on the susceptance of the transmission lines for the complete bipartite network, and hence, is valid for every pair of the tripped line $\hat e$ and the monitored line $e$ in different sub-grids. In practice, the actual LODF under the complete bipartite interface network is usually much lower than the theoretical bound due to the internal connectivity of the network. Therefore, the complete bipartite network can provide strong failure localization.

We remark that the complete bipartite interface network can be designed not only to eliminate the impact outside the sub-grid where the failure happens, but to maintain the same level of robustness within the sub-grid. Specifically, as stated in the following theorem, the LODF remains the same as the original network if the lines are in the same sub-grids, while the LODF is zero if the lines are in different sub-grids. 

\begin{thm}
Consider a network $\calG$ consisting of two sub-grids $\calG_1, \calG_2$ joined by two buses $s$ and $t$, and the modified network with the $2\times 2$ complete bipartite interface network. Suppose the effective susceptances between buses $s$ and $t$ for the two sub-grids $\calG_1, \calG_2$ is $b_1^{(e)}$ and $b_2^{(e)}$ respectively. If the susceptances of the lines in the complete bipartite network satisfies the following condition:
\begin{align*}
&b_{tt'} < \min (b_1,b_2) \text{ or } b_{tt'} > \max (b_1,b_2),\\
&b_{ss'} = \frac{b_1b_2}{b_{tt'}}, \ b_{st'} = \frac{b_2(b_1-b_{tt'})}{b_2-b_{tt'}}, \ b_{ts'} = \frac{b_1(b_2-b_{tt'})}{b_1-b_{tt'}},
\end{align*}
then we have
$$
K_{e,\hat e}^{(m)} = 
\begin{cases}
K_{e,\hat e}, \text{ if the lines $e,\hat e$ are in the same sub-grid,}\\
0, \text{ if the lines $e,\hat e$ are in different sub-grids.}
\end{cases}
$$
\end{thm}

\begin{proof}[Proof (sketch)]
This result can be proved using the fact that the effective susceptance between buses $(s,t)$ and $(s',t')$ remains the same if the conditions are satisfied. 
\end{proof}

\subsection{Comments}
The theoretical analysis presented in this section focuses on non-bridge line failures where post-contingency power injections are assumed to remain constant. In practice, however, the injections might change due to the real-time automatic controls of the power grid. The situation is even more complicated when islanding occurs due to bridge line failures. Even if a detailed model of these situations is beyond the scope of this paper, the three interface networks presented here are capable of localizing the impact of injection fluctuations, as it can be seen through a similar analysis of the PTDF. 

We remark that the sensitivity factors PTDF and LODF we considered in this paper are independent of the power injections and transmission line capacities and, in fact, are known to only depend on the topological structure of the power grid, see~\cite{guo2021line1}. Therefore, our analysis sheds light on how the power grid can be optimized by possibly re-designing the network through line switching or bus splitting for planning purposes. If the pre-contingency system state and the line capacities are known, our analysis on LODF can be helpful to identify possible successive failures, for which corrective control actions can then be performed.

\begin{figure*}[htbp]
	\centering
	\subfloat[Series interfacing network]{\includegraphics[width=0.3\textwidth]{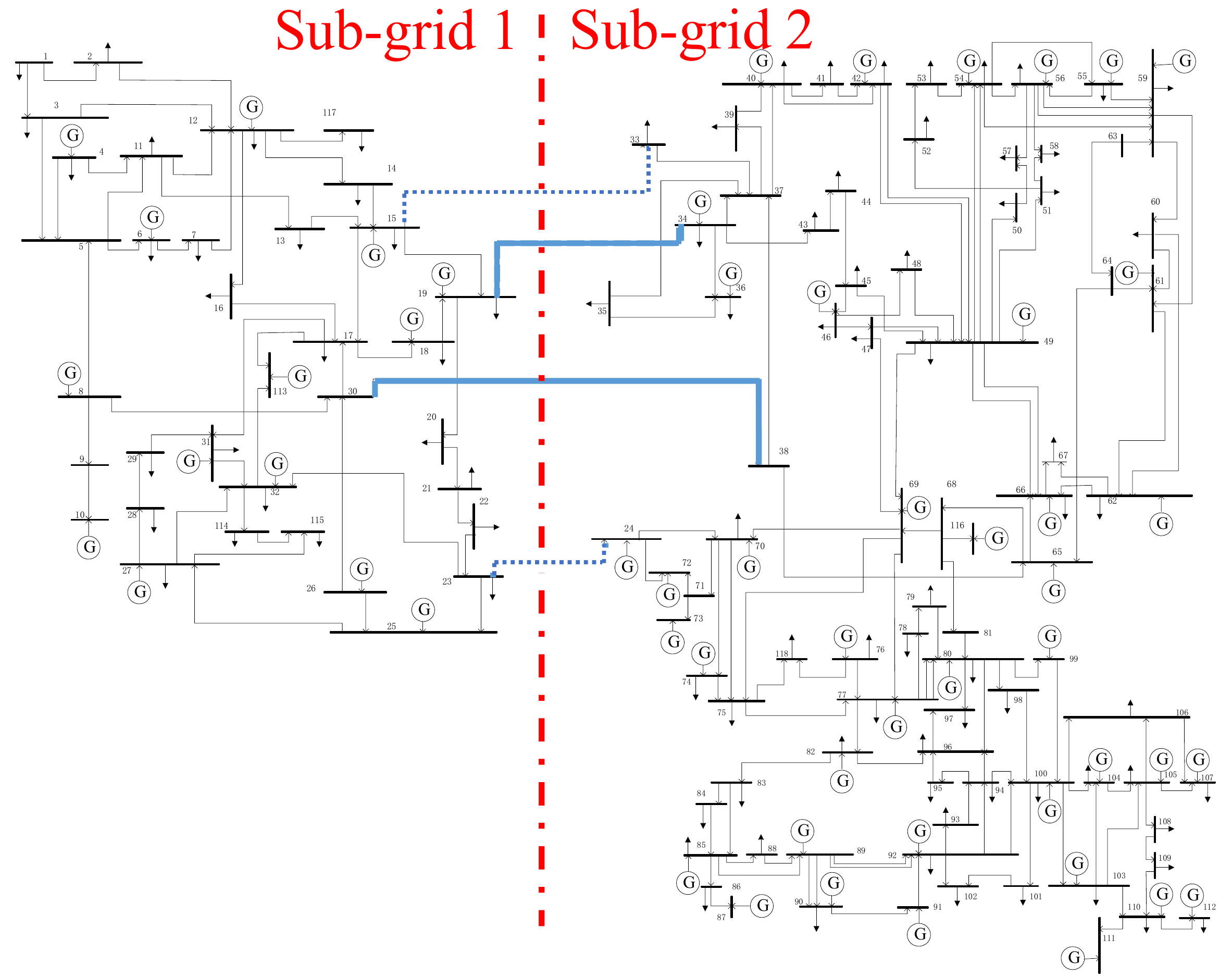} \label{fig_118_series}}
	\hfil
	\subfloat[Parallel interfacing network]{\includegraphics[width=0.3\textwidth]{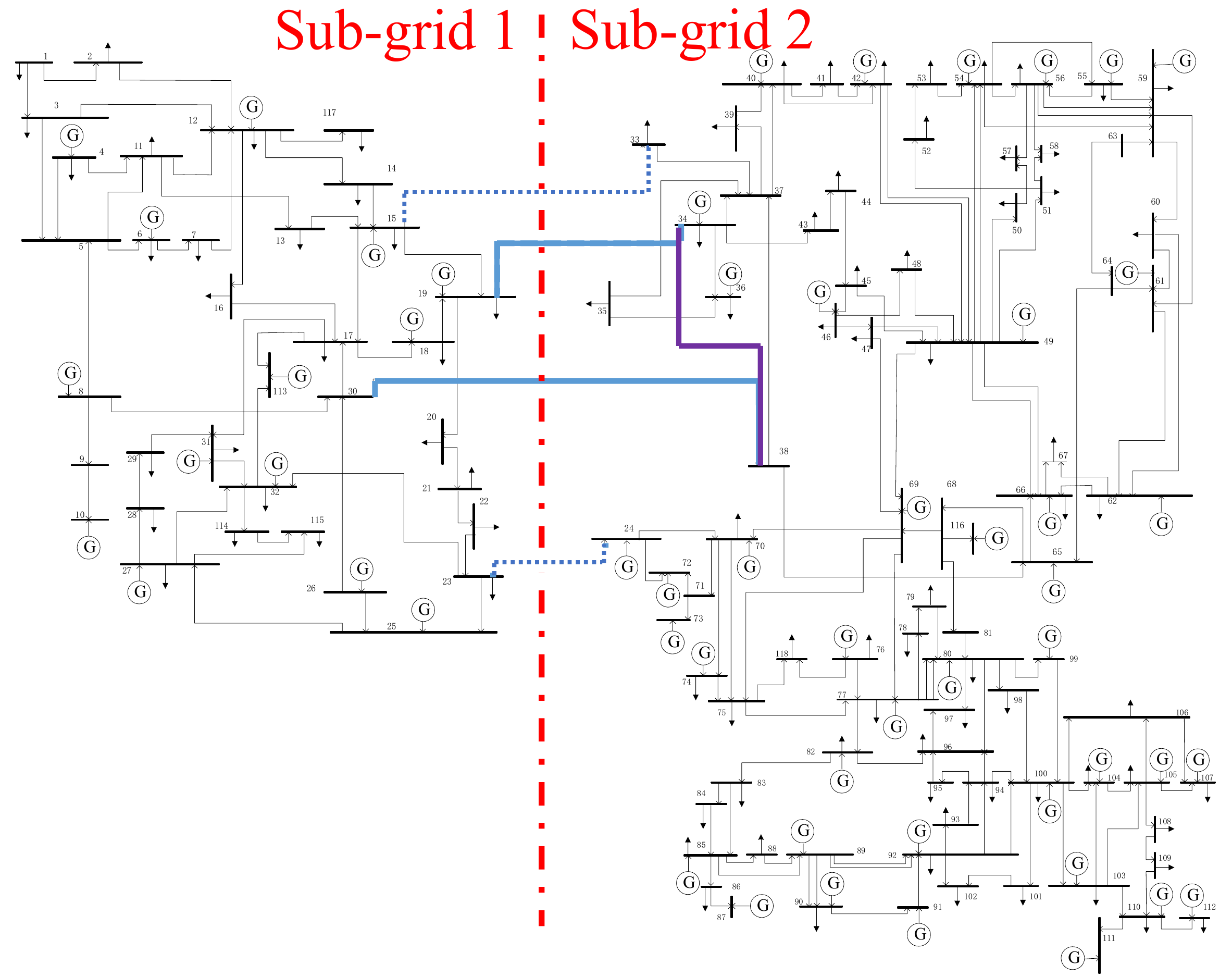} \label{fig_118_parallel}}
	\hfil
	\subfloat[Complete bipartite network]{\includegraphics[width=0.3\textwidth]{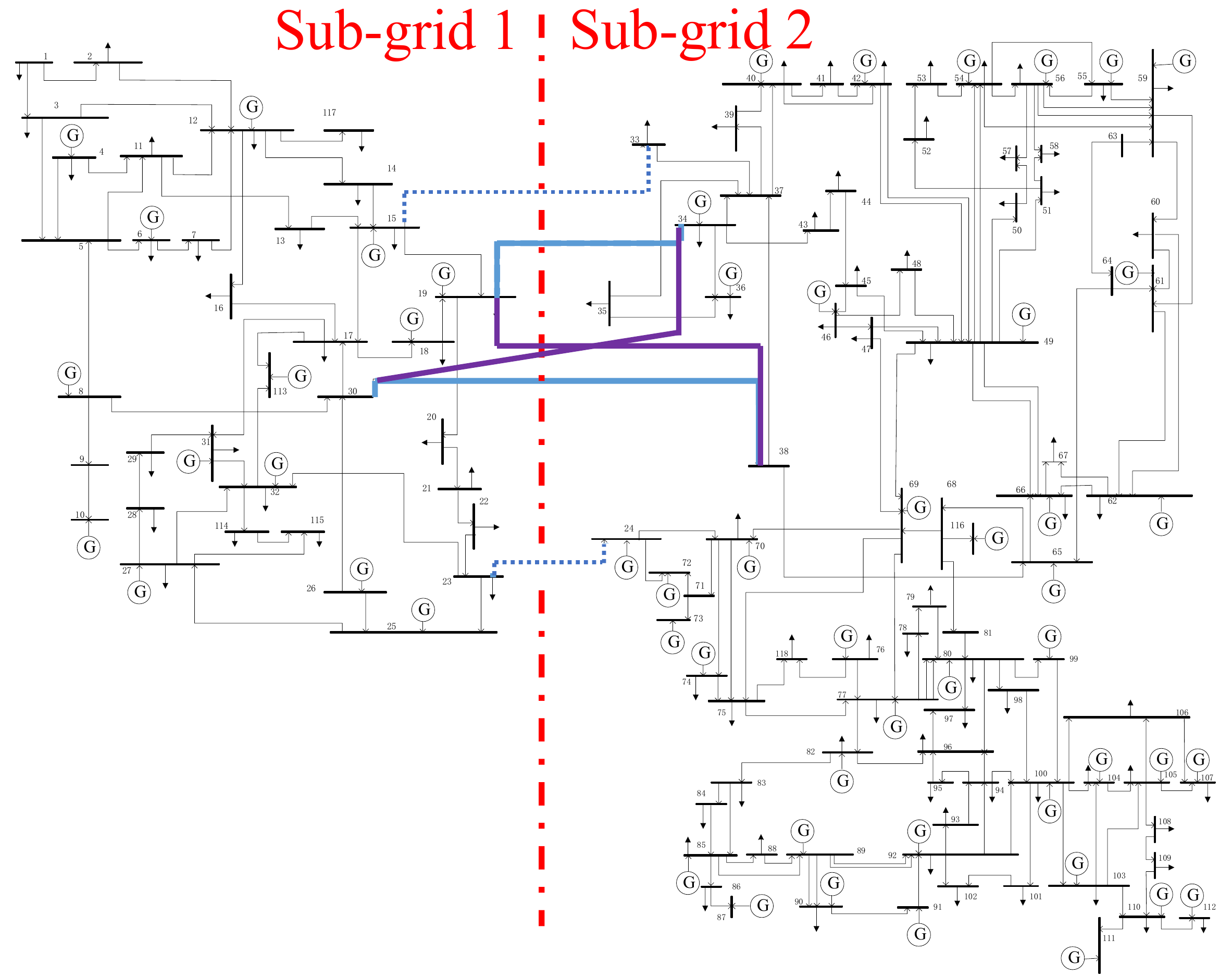} \label{fig_118_bipartite}}
	\caption{Two sub-grids are connected by 4 blue lines in the original IEEE 118-bus network. (a) Two dashed lines are switched off to create the series interface network. (b) One purple line is added to create the parallel interface network. (c) Two purple lines are added to create the complete bipartite network.}
	\label{fig_118}
\end{figure*}

\section{Case Study}\label{sec:simulation}

In this section, we evaluate the failure localization performance of the three interface networks studied in the previous section for the IEEE 118-bus test network. We start with the DC model, and then extend it to the AC model. 

\subsection{Experimental Setup}

We split the IEEE 118-bus test network network into two sub-grids connected by four tie-lines as shown in Fig.~\ref{fig_118}. Note that the sub-grids are not connected by the cut vertices as in Fig.~\ref{fig_first_case}. Therefore we modify the tie-lines connecting the sub-grids to create interface networks as follows. For a series interface network, we switch off the two dashed blue lines and keep the two solid blue lines as in Fig.~\ref{fig_118_series}. A parallel interface network is built on top of the series network, where we add the purple line as in Fig.~\ref{fig_118_parallel}. The complete bipartite interface network is achieved by connecting the end-points of solid blue lines as the solid purple lines in Fig.~\ref{fig_118_bipartite}. 
We then calculate the LODF as a metric to quantify the failure impact for each interface network.

\begin{figure}[!h]
	\centering
	\subfloat[DC LODF for lines in different sub-grids.]{\includegraphics[width=0.74\columnwidth]{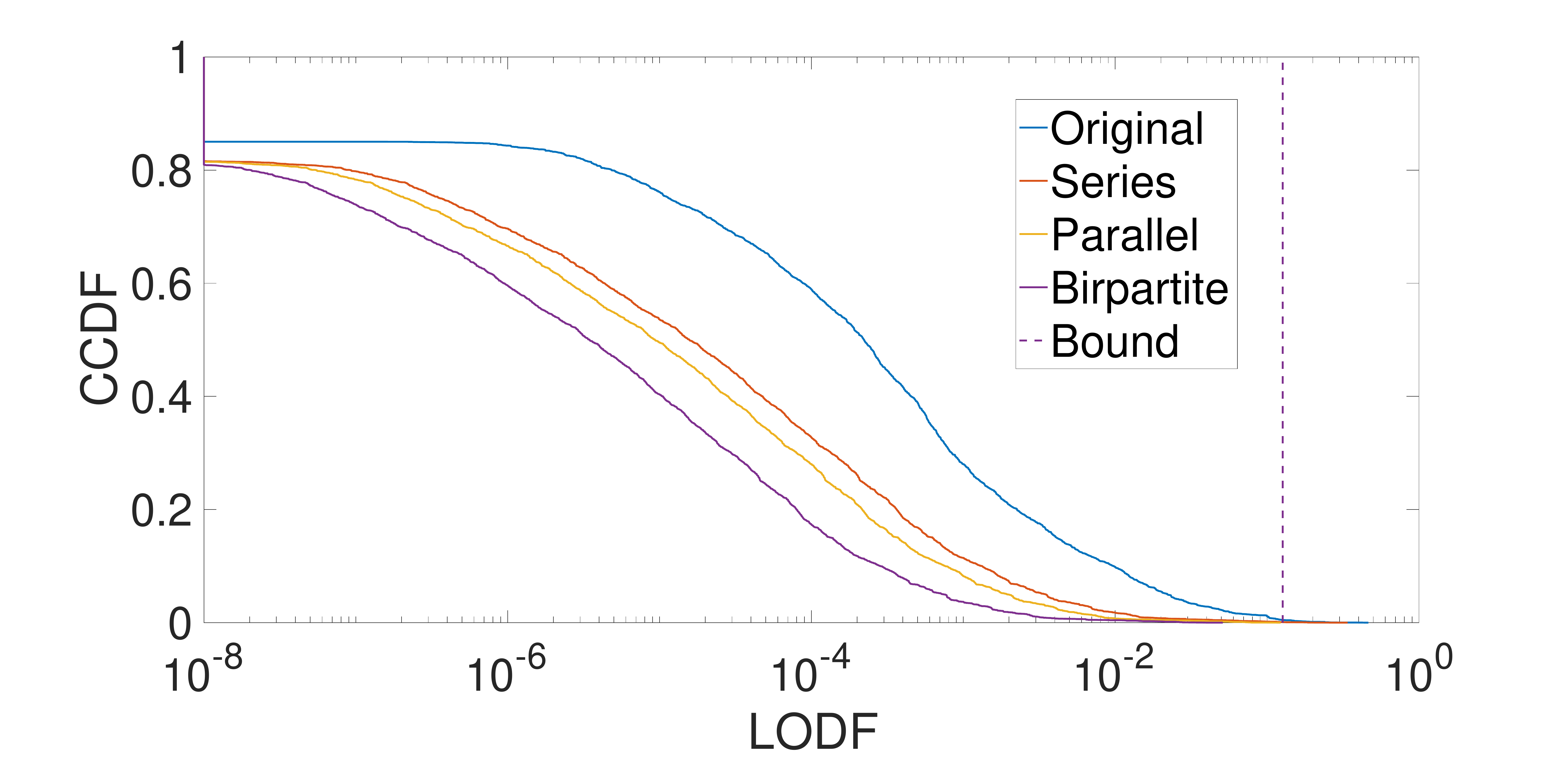} \label{fig_LODF_cross}}\\
	\subfloat[AC LODF for lines in different sub-grids.]{\includegraphics[width=0.74\columnwidth]{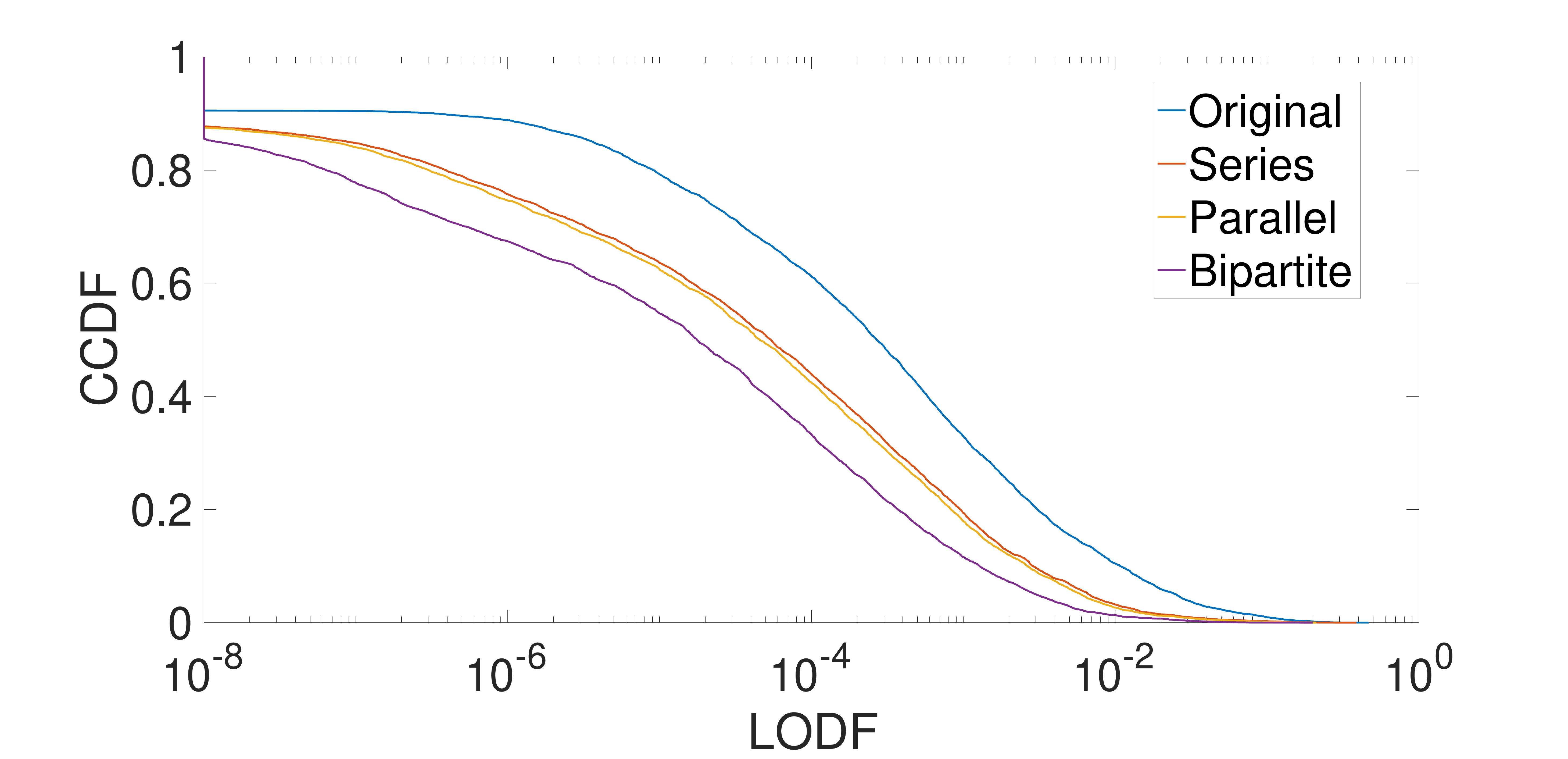} \label{fig_LODF_cross_AC}}\\
	\subfloat[DC LODF for lines in the same sub-grid.]{\includegraphics[width=0.74\columnwidth]{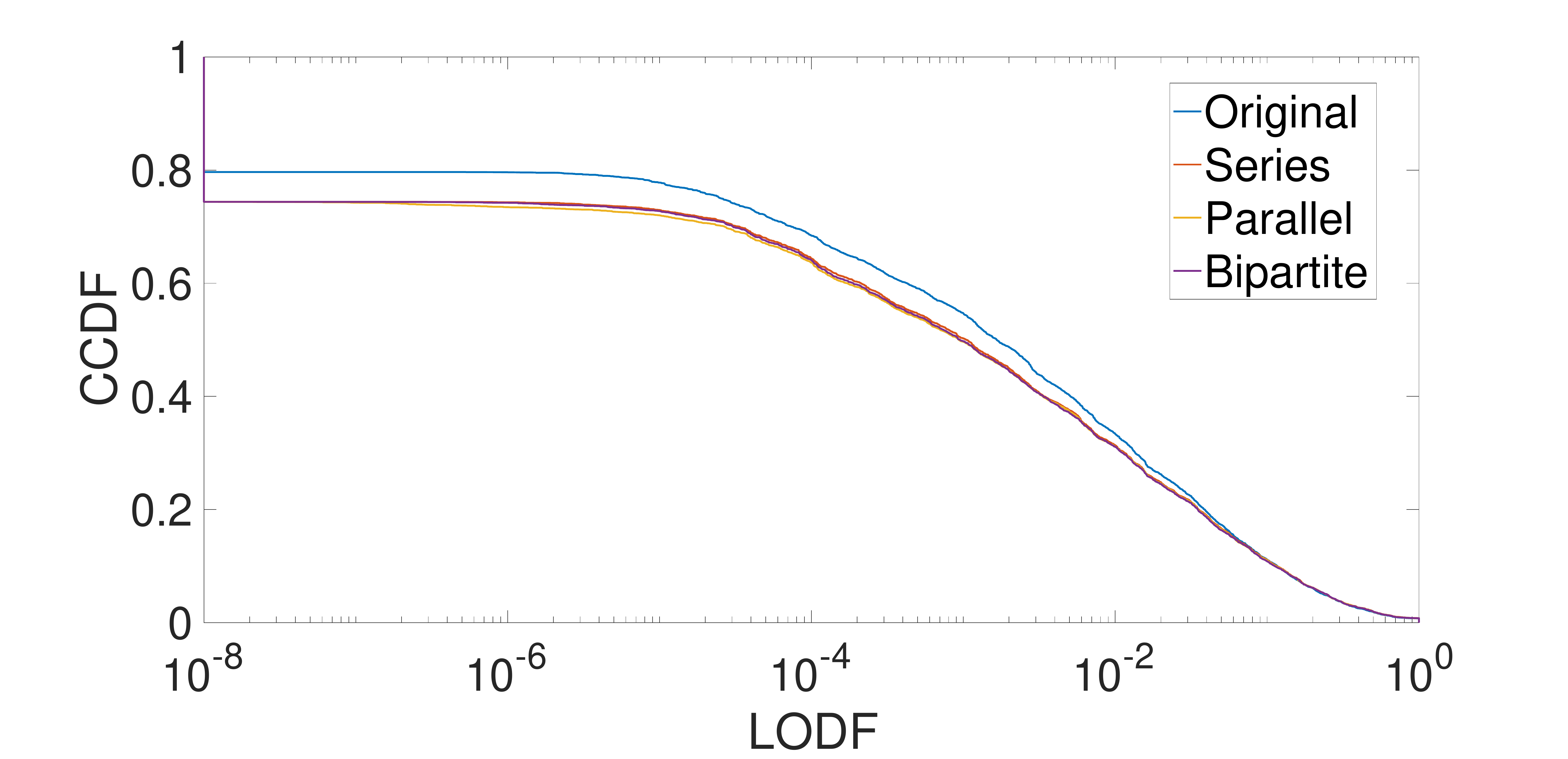} \label{fig_LODF_same}}\\
	\subfloat[AC LODF for lines in the same sub-grid.]{\includegraphics[width=0.74\columnwidth]{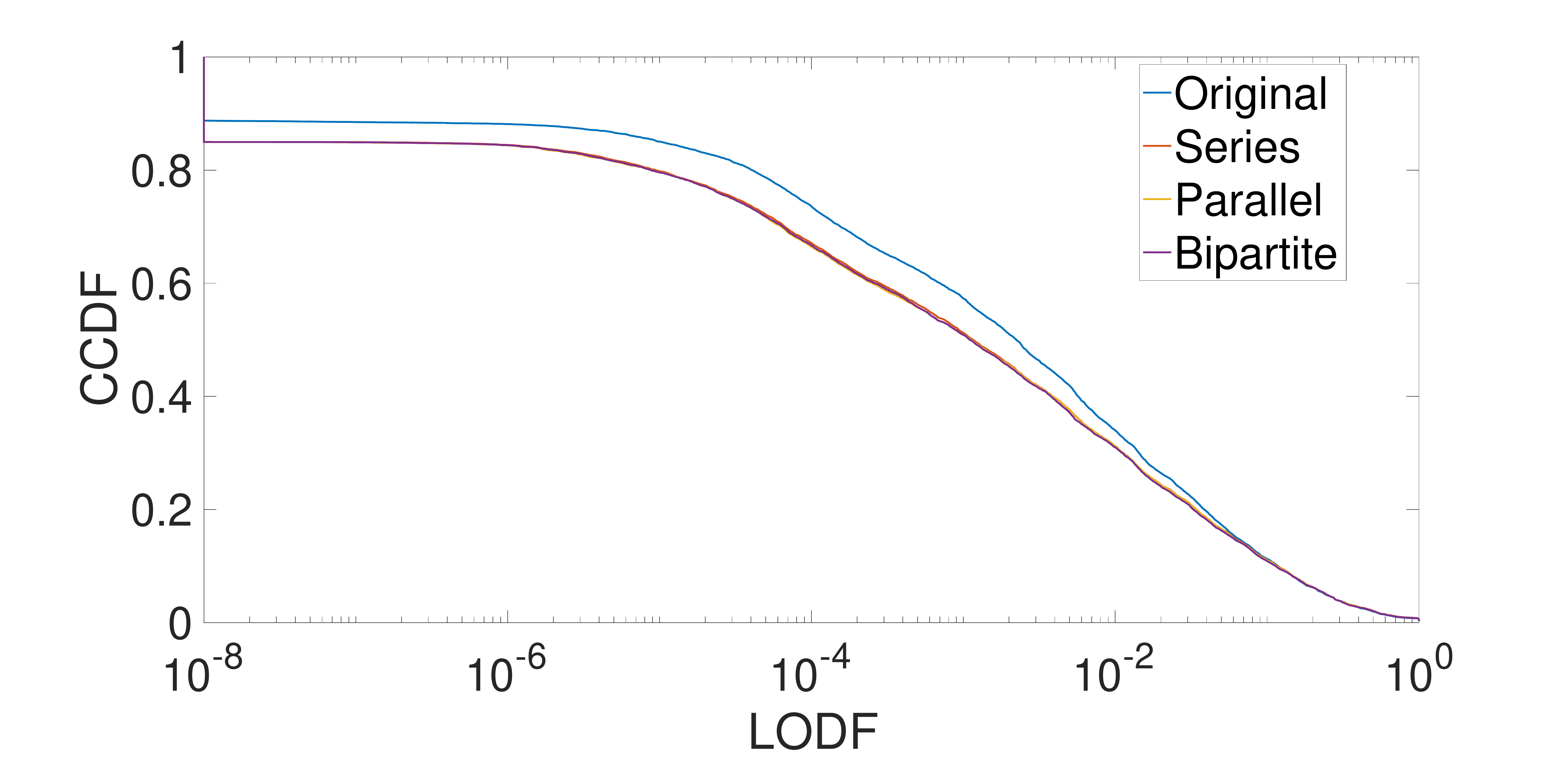} \label{fig_LODF_same_AC}}
	\caption{The CCDF of LODF for monitored line $e$ and tripped line $\hat e$ under DC (a,c) and AC (b,d) models. (a,b) $e,\hat e$ are in different sub-grids. (c,d) $e,\hat e$ are in the same grid.}
	\label{fig_sim_dc}
\end{figure}

\subsection{Experimental Results}

\paragraph{DC Model} We start with evaluating failure localization under the DC power flow model. The DC LODF is well-defined and can be computed as~\eqref{eqn:LODF} if the tripped line does not disconnect the network. 

We first compare the failure localizability across the sub-grids under various interface networks. Specifically, we compute the LODF for all pairs of tripped lines $e$ and monitored lines $\hat e$ in different sub-grids and demonstrate the complementary cumulative distribution function (CCDF) of the absolute LODF in Fig.~\ref{fig_LODF_cross}. Note that the x-axis is in logarithmic scale and we set the LODF $|K_{e,\hat e}|\leq 10^{-8}$ as zero. The vertical dashed line represents the theoretical bound of the LODF for the complete bipartite network.
We observe that all three interface networks reduce the LODF across the sub-grids. For the original 118-bus network, there are roughly $10\%$ pairs of lines with the absolute LODF greater than 0.01, while those cases are negligible ($1\%$) with the series interface network. As expected, adding a parallel interface network on top of the series network further decreases the LODF. The complete bipartite interface network achieves the best localization performance, even though the susceptance does not satisfy the rank-1 condition to completely localize the failure within the sub-grid, i.e. $b_{ss'}b_{tt'} \neq b_{st'}b_{ts'}$.

It is crucial to analyze the impact within the same sub-grid where the line failure happens as well. In Fig.~\ref{fig_LODF_same}, we show the CCDF of the absolute LODF for the pairs of tripped line and monitored line within the same sub-grid. We observe that the distributions of LODF within the sub-grid for the series, parallel and complete bipartite interface networks are very similar, all lower than the original network. Therefore, introducing the proposed interface networks properly will not decrease the robustness for the sub-grids against failures.

We remark that the two sub-grids of IEEE 118-bus network does not follow the definition of original network in Section~\ref{sec:localization}: they are connected by four tie-lines instead of only two buses. Nevertheless, the LODF for the modified networks with all three interface network decreases. It suggests a broader range of applicability and stronger failure localizability for the interface networks. This, however, requires a proper selection on which transmission lines to keep, and we leave it as a future direction to explore.

\paragraph{AC model}

We further evaluate the localization performance under AC model. Since there is no closed-form expression for AC LODF, we calculate the LODF directly using the definition.
Specifically, we adopt the line parameters and the nominal power injections from~\cite{babaeinejadsarookolaee2019power} as the pre-contingency operating status. For every non-bridge transmission line $\hat e$, we compute the post-contingency flow with AC power flow equations when line $\hat e$ trips, assuming that the post-contingency injections remain the same. The LODF is thus computed as $K_{e,\hat e}=\frac{\Delta f_e}{f_{\hat e}}$, where $\Delta f_e$ is the flow change over line $e$ and $f_{\hat e}$ is the pre-contingency flow over line $\hat e$.

The CCDF of LODF for all pairs of the monitored line and the tripped line are shown in Fig.~\ref{fig_LODF_cross_AC} and Fig.~\ref{fig_LODF_same_AC}. We notice that the network in which the sub-grids are connected by any of the three interface networks achieves higher failure localizability similarly to the DC model. It should be noted that the LODF is not zero for the complete bipartite network under the AC model, even when the susceptances are designed to satisfy the rank-1 condition. Nevertheless, all interface networks reduce failure impact across sub-grids, while maintaining similar robustness within the sub-grid.


\section{Conclusion}\label{sec:con}
In this paper, we propose three interface networks connecting sub-grids to achieve stronger failure localization while maintaining robust network connectivity. Both theoretical analysis and case studies validate our proposed method. There are a number of important directions for future exploration of this topic. The most important and challenging extension is to consider larger interface networks. In this paper we have considered $2\times2$ interface networks, but larger networks have the potential to provide even more robust connections between sub-grids. However, it is quite challenging to characterize the LODF for larger interface networks without assuming very specific topological properties and thus to ensure localization of failures for such interface networks.



\bibliographystyle{IEEEtran}
\bibliography{biblio}

\end{document}